\documentclass[10pt,a4paper]{article}
\usepackage{url}
\usepackage{verbatim}
\usepackage{latexsym}
\usepackage{amssymb,amstext,amsmath,amsthm}
\usepackage{epsf}
\usepackage{epsfig}
\usepackage{a4wide}
\usepackage{verbatim}
\usepackage{proof}
\usepackage{latexsym}
\newtheorem{theorem}{Theorem}[section]

\usepackage{float}
\floatstyle{boxed}
\restylefloat{figure}

\def\oge{\leavevmode\raise
.3ex\hbox{$\scriptscriptstyle\langle\!\langle\,$}}
\def\feg{\leavevmode\raise
.3ex\hbox{$\scriptscriptstyle\,\rangle\!\rangle$}}

\newcommand{\UU}{U}
\newcommand{\WW}{\mathsf{W}}
\newcommand{\MM}{\mathsf{M}}
\newcommand{\GU}{{\cal U}}
\newcommand{\GV}{{\cal V}}

\newcommand{\Elem}{\mathsf{Elem}}

\newcommand{\Val}[1]{[#1]}
\newcommand{\id}{1}
\newcommand{\pp}{\mathsf{p}}
\newcommand{\qq}{\mathsf{q}}

\newcommand{\lift}[1]{\overline{#1}}

\def\N0{\hbox{\sf N}_0}

\newcommand{\app}{\mathsf{app}}
\newcommand{\APP}[2]{\mathsf{app}(#1,#2)}

\newcommand{\brec}{\hbox{\sf{brec}}}

\newcommand{\conv}{\mathsf{conv}}
\newcommand{\Term}{\mathsf{Term}}
\newcommand{\Neut}{\mathsf{Neut}}

\newcommand{\Norm}{\mathsf{Norm}}
\newcommand{\Var}{\mathsf{Var}}

\newcommand{\Type}{\mathsf{Type}}

\newcommand{\del}[1]{}

\newcommand{\subst}[1]{{\langle}#1{\rangle}}

\begin{document}

\title{Canonicity and normalization for Dependent Type Theory}

\author{Thierry Coquand}
\date{Computer Science Department, University of Gothenburg}
\maketitle

%\rightfooter{}

\section*{Introduction}

 We show canonicity and normalization for dependent type theory 
with a cumulative sequence of universes $\UU_0:\UU_1\dots$ 
{\em with} $\eta$-conversion. We give the argument in a constructive set theory 
CZFu$_{<\omega}$, designed by P. Aczel \cite{Aczel}. 
We provide a purely algebraic presentation of a canonicity proof, as a way
to build new (algebraic) models of type theory. 
We then present a normalization
proof, which is technically more involved, but is based on the same idea.
%This version can be seen as a direct
%generalization of the canonicity proof for G\"odel system T presented in \cite{Shoenfield}.
We believe our argument to be a simplification of existing proofs \cite{ML72,ML73,Abel,Coq}, in the sense
that we never need to introduce a reduction relation, and the proof theoretic strength
of our meta theory is as close as possible to the one of the object theory \cite{Aczel,Rathjen}. 

 Let us expand these two points. If we are only interested in {\em canonicity}, i.e. to prove
that a closed Boolean is convertible to $0$ or $1$, one argument for simple type theory
(as presented e.g. in \cite{Shoenfield}) consists in defining a ``reducibility''\footnote{The terminology
for this notion seems to vary: in \cite{Godel}, where is was first introduced, it is called
``berechenbarkeit'', which can be translated by ``computable'',
in \cite{Tait} it is called ``convertibility'', and in \cite{Shoenfield} it is
called ``reducibility''.} predicate by induction
on the type. For the type of Boolean, it means exactly to be convertible to $0$ or $1$, 
and for function types, it means that it sends a reducible argument to a reducible value. It is
then possible to show by induction on the typing relation that any closed term is reducible.
In particular, if this term is a Boolean, we obtain canonicity. The problem of extending this
argument for a dependent type system with universes is in the definition of what should be
the reducibility predicate for universes. It is natural to try an inductive-recursive definition;
this was essentially the way it was done in \cite{ML72}, which is an early instance of 
an inductive-reductive definition. We define when an element of the universe
is reducible, and, by induction on this proof, what is the associated reducibility predicate
for the type represented by this element.
However, there is a difficulty in this approach: it might 
well be {\em a priori} that an element is both convertible for instance to the type
of Boolean or of a product type, and if this is the case, the previous inductive-recursive definition is ambiguous. 

 In \cite{ML72}, this problem is solved by considering first a 
reduction relation, and then showing this reduction relation to be confluent, and defining convertibility
as having a commun reduct. This does {\em not} work however
when conversion is defined as a {\em judgement} (as in \cite{ML73,Abel}).
This is an essential difficulty, and a relatively subtle and complex argument is involved in 
\cite{Abel,Coq}
to solve it: one defines first an {\em untyped} reduction relation and a reducibility {\em relation}, which is
used first to establish a confluence property.

 The main point of this paper is that this essential difficulty can be solved, in a seemingly magical way, 
by considering {\em proof-relevant} reducibility, that is where reducibility is defined  as a {\em structure} and not only
as a {\em property}. Such an approach is hinted in the reference \cite{ML73}, but \cite{ML73} still
introduces a reduction relation, and also presents a version of type theory with a restricted
form of conversion (no conversion under abstraction, and no $\eta$-conversion;
this restriction is motivated in \cite{ML74}). 

Even for the base type, reducibility is a structure: 
the reducibility structure of
an element $t$ of Boolean type contains either $0$ (if $t$ and $0$ are convertible)
or $1$ (if $t$ and $1$ are convertible) and  this might a priori contains both $0$ and $1$.
Another advantage of our approach, when defining reducibility in a proof-relevant way, is that
the required meta-language is weaker than the one used for a reducibility relation (where one has to do proofs by induction on
this reducibility relation).

 Yet another aspect that was not satisfactory in previous attempts \cite{Abel,Coq} is that it involved
essentially a {\em partial equivalence relation model}. One expects that this would be needed
for a type theory with an extensional equality, but not for the present version of type theory.
This issue disappears here: we only consider {\em predicates}
(that are proof-relevant).

 A more minor contribution of this paper is its {\em algebraic} character. For both
canonicity and decidability of conversion, one considers first a general model construction
and one obtains then the desired result by instantiating this general construction to the
special instance of the initial (term) model, using in both cases only the abstract characteristic
property of the initial model.

\section{Informal presentation}

 We first give an informal presentation of the canonicity proof by first expliciting
the rules of type theory and then explaining the  reducibility argument,

\subsection{Type system}

 We use conversion as judgements \cite{Abel}. Note that it is not clear
a priori that subject reduction holds.

$$
\frac{\Gamma\vdash A:U_n}{\Gamma,x:A\vdash}~~~~~~\frac{}{()\vdash}~~~~~~~
\frac{\Gamma\vdash}{\Gamma\vdash x:A}~(x\!:\! A~in~\Gamma)$$
$$
%\frac{\Gamma,x:A\vdash B}{\Gamma\vdash \Pi (x:A) B}~~~~~~~~~
\frac{\Gamma\vdash A:\UU_n~~~~~~\Gamma,x:A\vdash B:\UU_n}{\Gamma\vdash \Pi (x:A) B:\UU_n}~~~~~~~~~
\frac{\Gamma,x:A\vdash t:B}{\Gamma\vdash \lambda (x:A) t:\Pi (x:A) B}~~~~~~~~
\frac{\Gamma\vdash t:\Pi (x:A) B~~~~~~\Gamma\vdash u:A}
     {\Gamma\vdash t~u:B(u)}
$$
%$$
%\frac{\Gamma\vdash A:\UU_n}{\Gamma\vdash A}
%$$
$$
\frac{\Gamma\vdash A:\UU_n}{\Gamma\vdash A:\UU_m}~(n\leqslant m)~~~~~~
\frac{}{\Gamma\vdash \UU_n:\UU_m}~(n<m)~~~~~
\frac{}{\Gamma\vdash N_2:\UU_n}
$$

 The conversion rules are
$$
\frac{\Gamma\vdash t:A~~~~~~\Gamma\vdash A~ \conv~ B:U_n}{\Gamma\vdash t:B}~~~~~~~~~
\frac{\Gamma\vdash t ~\conv~u:A~~~~~~\Gamma\vdash A  ~\conv~ B:U_n}{\Gamma\vdash t ~\conv~u:B}
$$
$$
\frac{\Gamma\vdash t:A}{\Gamma\vdash t ~\conv~t:A}~~~~~~~~~
\frac{\Gamma\vdash t ~\conv~v:A~~~~~~~~~\Gamma\vdash u ~\conv~v:A}{\Gamma\vdash t ~\conv~u:A}
$$
$$
\frac{\Gamma\vdash A ~\conv~B:\UU_n}{\Gamma\vdash A ~\conv~B:U_m}~(n\leqslant m)~~~~~~
\frac{\Gamma\vdash A_0  ~\conv~ A_1:\UU_n~~~~~~~~\Gamma,x:A_0\vdash B_0  ~\conv~ B_1:\UU_n}
     {\Gamma\vdash \Pi (x:A_0) B_0  ~\conv~ \Pi (x:A_1) B_1:\UU_n}
$$
$$
\frac{\Gamma\vdash t ~\conv~t':\Pi (x:A) B~~~~~~\Gamma\vdash u:A}
     {\Gamma\vdash t~u ~\conv~t'~u:B(u)}~~~~~~~~~~~
\frac{\Gamma\vdash t:\Pi (x:A) B~~~~~~\Gamma\vdash u  ~\conv~ u':A}
     {\Gamma\vdash t~u  ~\conv~ t~u':B(u)}
$$
$$
\frac{\Gamma,x:A\vdash t:B~~~~~~~~\Gamma\vdash u:A}{\Gamma\vdash (\lambda (x:A) t)~u  ~\conv~ t(u):B(u)}
$$

 We consider type theory with $\eta$-rules
$$
\frac{\Gamma\vdash t:\Pi (x:A) B~~~~\Gamma\vdash u:\Pi (x:A) B~~~~\Gamma,x:A\vdash t~x  ~\conv~ u~x:B}
     {\Gamma\vdash t  ~\conv~ u:\Pi (x:A) B}$$

 Finally we add $N_2:\UU_1$ with the rules
$$
\frac{}{\Gamma\vdash 0:N_2}~~~~~~~~~~\frac{}{\Gamma\vdash 1:N_2}
~~~~~~~~~
\frac{\Gamma,x:N_2\vdash C:U_n~~~~~\Gamma\vdash a_0:C(0)~~~~~~~\Gamma\vdash a_1:C(1)}
     {\Gamma\vdash \brec~(\lambda x.C)~a_0~a_1:\Pi(x:N_2)C}
$$
with computation rules
${\brec~(\lambda x.C)~a_0~a_1~0  ~\conv~ a_0:C(0)}$
and
${\brec~(\lambda x.C)~a_0~a_1~1  ~\conv~ a_1:C(1)}$.

\subsection{Reducibility proof}

 The informal reducibility proof consists in associating to each closed expression
$a$ of type theory (treating equally types and terms) an abstract object $a'$
which represents a ``proof'' that $a$ is reducible. If $A$ is a (closed) type, then
$A'$ is a family of sets over the set $\Term(A)$ of closed expressions of type $A$
{\em modulo conversion}. If $a$ is
of type $A$ then $a'$ is an element of the set $A'(a)$.

 The metatheory is a (constructive) set theory with a commulative hierarchy of universes
${\cal U}_n$ \cite{Aczel}.

 This is defined by structural induction on the expression as follows
\begin{itemize}
\item $(c~a)'$ is $c'~a~a'$
\item $(\lambda (x:A) b)'$ is the function which takes as arguments
a closed expression $a$ of type $A$  and an element $a'$ in $A'(a)$
and produces $b'(a,a')$
\item $(\Pi (x:A)B)'(w)$ for $w$ closed expression of type $\Pi (x:A)B$
is the set $\Pi (a:\Term(A))(a':A'(a))B'(a,a')(w~a)$
\item $N_2'(t)$ is the set $\{0~|~t~\conv~0\}\cup\{1~|~t~\conv~1\}$
\item $U_n'(A)$ is the set $\Term(A)\rightarrow {\cal U}_n$
\end{itemize}

 It can then be shown\footnote{We prove this statement by induction on the derivation and consider
a more general statement involving a context; we don't provide the details
in this informal part since this will be covered in the next section.}
that if $a:A$ then $a'$ is an element of $A'(a)$
and furthermore that if $a~\conv~b:A$ then $a' = b'$ in $A'(a) = A'(b)$.
In particular, if $a:N_2$ then $a'$ is $0$ or $1$ and we get that $a$
is convertible to $0$ and $1$.

\medskip

 One feature of this argument is that the required 
meta theory, here constructive set theory, is known to be of similar strength as the corresponding
type theory; for a term involving $n$ universes, the meta theory will need $n+1$ universes \cite{Rathjen}.
This is to be contrasted with the arguments in \cite{ML72,Abel,Coq} involving induction recursion
which is a much stronger principle.

\medskip

 We believe that the mathematical purest way to formulate this argument
is an algebraic argument, giving a (generalized) algebraic presentation of type theory.
We then use only of the {\em term} model the fact that it is the {\em initial} model
of type theory. This is what is done in the next section.

\section{Model and syntax of dependent type theory with universes}

\subsection{Cumulative categories with families}

 We present a slight variation (for universes) of the notion of ``category'' with families
\cite{Dybjer}\footnote{As emphasized in this reference, these models should be more exactly
thought of as {\em generalized algebraic structures} rather than {\em categories}; e.g. the initial model 
is defined up to isomorphism and not up to equivalence). This provides a generalized algebraic notion
of model of type theory.}. 
A model is given first by a class of {\em contexts}. If $\Gamma,\Delta$ are two given contexts
we have a set $\Delta\rightarrow\Gamma$ of {\em substitutions} from $\Delta$ to $\Gamma$.
These collections of sets are equipped with operations that
satisfy the laws of composition in a category: we have a substitution $\id$ 
in $\Gamma\rightarrow\Gamma$ and
a composition operator $\sigma\delta$ in $\Theta\rightarrow\Gamma$ if
$\delta$ is in $\Theta\rightarrow\Delta$ and $\sigma$ in $\Delta\rightarrow\Gamma$. Furthermore
we should have $\sigma \id = \id \sigma = \sigma$ and 
$(\sigma\delta)\theta = \sigma(\delta\theta)$ if $\theta:\Theta_1\rightarrow\Theta$.

 We assume to have a ``terminal'' context $()$: for any other context, there is a 
unique substitution, also written $()$, in $\Gamma\rightarrow ()$. In particular
we have $()\sigma = ()$ in $\Delta\rightarrow ()$ if $\sigma$ is in 
$\Delta\rightarrow \Gamma$. 

We write $|\Gamma|$ the set of substitutions $()\rightarrow\Gamma$.

\medskip

If $\Gamma$ is a context we have a cumulative sequence of sets $\Type_n(\Gamma)$
of {\em types over} $\Gamma$ at level $n$ (where $n$ is a natural number).
%We write $\Gamma\vdash A$ to express that $A$ is a type over $\Gamma$.
If $A$ in $\Type_n(\Gamma)$ and $\sigma$ in $\Delta\rightarrow\Gamma$ 
we should have $A\sigma$ in $\Type_n(\Delta)$.
Furthermore $A\id = A$ and $(A\sigma)\delta = A(\sigma\delta)$.
If $A$ in $\Type_n(\Gamma)$ we also have a collection $\Elem(\Gamma,A)$
of {\em elements of type} $A$.
If $a$ in $\Elem(\Gamma,A)$
and $\sigma$ in $\Delta\rightarrow\Gamma$ 
we have $a\sigma$ in $\Elem(\Delta,A\sigma)$. Furthermore
$a\id = a$ and $(a\sigma)\delta = a(\sigma\delta)$.
If $A$ is in $\Type_n()$ we write $|A|$ the set $\Elem((),A)$. 

 We have a {\em context extension operation}: if $A$ is in $\Type_n(\Gamma)$ then we can
form a new context $\Gamma.A$. Furthermore there is a projection
$\pp$ in $\Gamma.A\rightarrow \Gamma$ and a special element
$\qq$ in $\Elem(\Gamma.A,A\pp)$. If $\sigma$ is in $\Delta\rightarrow \Gamma$ and
$A$ in $\Type_n(\Gamma)$ and $a$ in $\Elem(\Delta,A\sigma)$ we have
an extension operation $(\sigma,a)$ in $\Delta\rightarrow \Gamma.A$.
We should have $\pp (\sigma,a) = \sigma$ and $\qq (\sigma,a) = a$ and
$(\sigma,a)\delta = (\sigma\delta,a\delta)$ and $(\pp,\qq) = \id$.

 If $a$ is in $\Elem(\Gamma,A)$ we write $\subst{a}= (\id,a)$ in $\Gamma\rightarrow \Gamma.A$.
Thus if $B$ is in $\Type_n(\Gamma.A)$ and $a$ in $\Elem(\Gamma,A)$ 
we have $B\subst{a}$ in $\Type_n(\Gamma)$.
If furthermore $b$ is in $\Elem(\Gamma.A,B)$ we have $b\subst{a}$ in $\Elem(\Gamma,B\subst{a})$. 

\medskip

 A {\em global} type of level $n$ is given by a an element $C$ in $\Type_n()$.
We write simply $C$ instead of $C()$ in $\Type_n(\Gamma)$ for $()$ in $\Gamma\rightarrow ()$.
Given such a global element $C$, a global element of type $C$ is given by 
an element $c$ in $\Elem((),C)$. We then write similarly
simply $c$ instead of $c()$ in $\Elem(\Gamma,C)$.

Models are sometimes presented by giving a class of special maps (fibrations), where a type
are modelled by a fibration and elements by a section of this fibration. In our case, the fibrations
are the maps $\pp$ in $\Gamma.A\rightarrow \Gamma$, and the sections of these fibrations
correspond exactly to elements in $\Elem(\Gamma,A)$.
Any element $a$ $\Elem(\Gamma,A)$ defines a section $\subst{a} = (\id,a):\Gamma\rightarrow\Gamma.A$
and any such section is of this form.

\subsection{Dependent product types}

 A category with families has {\em product types}
if we furthermore have one operation $\Pi~A~B$ in
$\Type_n(\Gamma)$ for $A$ is in $\Type_n(\Gamma)$ and $B$ is in $\Type_n(\Gamma.A)$.
We should have $(\Pi~A~B)\sigma = \Pi~(A\sigma)~(B\sigma^+)$
where $\sigma^+ = (\sigma\pp,\qq)$.
We have an abstraction operation $\lambda b$ in  $\Elem(\Gamma,\Pi~A~B)$ given
$b$ in $\Elem(\Gamma.A,B)$.
We have an application operation such that $\app(c,a)$ is in $\Elem(\Gamma,B\subst{a})$ if
$a$ is in $\Elem(\Gamma,A)$ and $c$ is in $\Elem(\Gamma,\Pi~A~B)$.
These operations should satisfy the equations
$$
\APP{\lambda b}{a} = b\subst{a}~~~~~~c = \lambda (\app~(c\pp,\qq))~~~~~
(\lambda b)\sigma = \lambda (b\sigma^+)~~~~
\APP{c}{a}\sigma = \APP{c\sigma}{a\sigma}
$$
where we write $\sigma^+ = (\sigma\pp,\qq)$.

\subsection{Cumulative universes}

 We assume to have global elements $U_n$ in $\Type_{n+1}(\Gamma)$
such that $\Type_n(\Gamma) = \Elem(\Gamma,U_n)$.

%% \medskip

%%  We can  also describe a model of type theory with {\em dependent  sums}. We should
%% have $\Sigma~A~B$ in $\Type_n(\Gamma)$ if $A$ is in $\Type_n(\Gamma)$ and $B$ in
%% $\Type_n(\Gamma.A)$.
%% If $\sigma$ is in $\Delta\rightarrow\Gamma$ 
%% we should have $(\Sigma~A~B)\sigma = \Sigma~(A\sigma)~(B\sigma^+)$.
%% If $a$ is in $\Elem(\Gamma,A)$
%% and $b$ in $\Elem(\Gamma,B\subst{a})$ we should have $(a,b)$ in $\Elem(\Gamma,\Sigma~A~B)$.
%% We require the equation $(a,b)\sigma = a\sigma,b\sigma$. We ask also for two operations
%% $\pp c$ in $\Elem(\Gamma,A)$ and 
%% $\qq c$ in $\Elem(\Gamma,B[\pp c])$ given $c$ in $\Elem(\Gamma,\Sigma~A~B)$
%% with the equations $\pp (a,b) = a$ and $\qq (a,b) = b$.

\subsection{Booleans}

 Finally we add the global constant $N_2$ in $\Type_0(\Gamma)$
and global elements
$0$ and $1$ in $\Elem(\Gamma,N_2)$.
Given $T$ in $\Type_n(\Gamma.N_2)$ and $a_0$ in $\Elem(\Gamma,T\subst{0})$
and $a_1$ in $\Elem(\Gamma,T\subst{1})$ we have an operation
$\brec(T,a_0,a_1)$ producing an element in $\Elem(\Gamma,\Pi~N_2~T)$
satisfying the equations
$\APP{\brec(T,a_0,a_1)}{0} = a_0$ and $\APP{\brec(T,a_0,a_1)}{1} = a_1$.

 Furthermore, $\brec(T,a_0,a_1)\sigma = \brec(T\sigma^+,a_0\sigma,a_1\sigma)$.

\section{Reducibility model}

 Given a model of type theory $\MM$  as defined above, we describe how to build a new
associated ``reducibility'' model $\MM^*$. 
When applied to the initial/term model $\MM_0$, this gives a proof of
canonicity which can be seen as a direct generalization of the argument presented
in \cite{Shoenfield} for G\"odel system T. As explained in the introduction, the main
novelty here is that we consider a proof-relevant notion of reducibility.

 A context of $\MM^*$ is given by a context $\Gamma$ of the model $\MM$ together with
a family of sets $\Gamma'(\rho)$ for $\rho$ in $|\Gamma|$.
A substitution in $\Delta,\Delta'\rightarrow^* \Gamma,\Gamma'$ is given by
a pair $\sigma,\sigma'$ with $\sigma$ in $\Delta\rightarrow\Gamma$ and
$\sigma'$ in $\Pi (\nu\in |\Delta|)\Delta'(\nu)\rightarrow \Gamma'(\sigma\nu)$.

 The identity substitution is the pair $1^* = 1,1'$ with $1'\rho\rho' = \rho'$.

 Composition is defined by $(\sigma,\sigma')(\delta,\delta') = \sigma\delta,(\sigma\delta)'$
with 
$$
(\sigma\delta)'\alpha\alpha' = \sigma'(\delta\alpha)(\delta'\alpha\alpha')
$$

\medskip

 The set $\Type^*_n(\Gamma,\Gamma')$ is defined to be the set of pairs
$A,A'$ where $A$ is in $\Type_n(\Gamma)$ and
$A'\rho\rho'$ is in $|A\rho|\rightarrow\GU_n$. We define then
$A'(\sigma,\sigma')\nu\nu' = A'(\sigma\nu)(\sigma'\nu\nu')$.

 We define $\Elem^*(\Gamma,\Gamma')(A,A')$ to be the set of pairs
$a,a'$ where $a$ is in $\Elem(\Gamma,A)$ and
$a'\rho\rho'$ is in $A'\rho\rho'(a\rho)$ for each $\rho$ in $|\Gamma|$
and $\rho'$ in $\Gamma'(\rho)$.
We define then $(a,a')(\sigma,\sigma') = a\sigma,a'(\sigma,\sigma')$
with $a'(\sigma,\sigma')\nu\nu' = a'(\sigma\nu)(\sigma'\nu\nu')$.

\medskip

 The extension operation is defined by $(\Gamma,\Gamma').(A,A') = \Gamma.A,(\Gamma.A)'$
where $(\Gamma.A)'(\rho,u)$ is the set of pairs $\rho',u'$
with $\rho'\in \Gamma'(\rho)$ and $u'$ in $A' \rho \rho'(u)$.

 We define an element $\pp^* = \pp,\pp'$ in $(\Gamma,\Gamma').(A,A')\rightarrow^* \Gamma,\Gamma'$
by taking $\pp'(\rho,u)(\rho',u') = \rho'$.
We have then an element $\qq,\qq'$ in $\Elem^*((\Gamma,\Gamma').(A,A'),(A,A')\pp^*)$
defined by $\qq'(\rho,u)(\rho',u') = u'$.

\subsection{Dependent product}

 We define a new operation $\Pi^*~(A,A')~(B,B') = \Pi~A~B,(\Pi~A~B)'$ where
$(\Pi~A~B)'\rho\rho'(w)$ is the set
$$
\Pi (u\in |A\rho|)\Pi (u'\in A'\rho\rho'(u))B'(\rho,u)(\rho',u')(\APP{w}{u})
$$

 If $b,b'$ is in $\Elem^*((\Gamma,\Gamma').(A,A'),(B,B'))$ then 
$\lambda^* (b,b') = \lambda b, (\lambda b)'$ where $(\lambda b)'$ is defined by the equation
$$
(\lambda b)' \rho\rho' u u' = b'(\rho,u)(\rho',u')
$$
which is in 
$$
B'(\rho,u)(\rho',u')(\APP{(\lambda b)\rho}{u}) = B'(\rho,u)(\rho',u')(b(\rho,u))
$$

 We have an application operation $\app^*((c,c'),(a,a')) = (\app(c,a),\app(c,a)')$
where $\app(c,a)'\rho\rho' = c'\rho\rho'(a\rho)(a'\rho\rho').$

\subsection{Universes}

 We define $\UU_n'(A)$ for $A$ in $|\UU_n|$ to be the set of functions
$|A|\rightarrow \GU_n$. Thus an element $A'$ of $\UU_n'(A)$ is a family
of sets $A'(u)$ in $\GU_n$ for $u$ in $|A|$. The universe $\UU_n^*$ of $\MM^*$
is defined to be the pair $\UU_n,\UU_n'$ and we have
$\Elem^*((\Gamma,\Gamma'),\UU_n^*) = \Type_n^*(\Gamma,\Gamma')$.

\subsection{Booleans}

 We define $N_2'(u)$ for $u$ in $|N_2|$ to be the set consisting of
$0$ if $u = 0$ and of $1$ if $u = 1$. We have $N_2'$ in $\UU_0'(N_2)$.
Note that $N_2'(u)$ may not be a subsingleton if we have $0 = 1$ in the model.
We define $\brec(T,a_0,a_1)'\rho\rho'u u'$ to be $a_0'\rho\rho'$ if $u' = 0$
and to be $a_1'\rho\rho'$ if $u' = 1$.

\subsection{Main result}

\begin{theorem}
The new collection of context, with the operations $\rightarrow^*,~\Type_n^*,\Elem^*$
and $U_n^*$ and $N_2^*$ define a new model of type theory.
\end{theorem}

 The proof consists in checking that the required equalities hold for the operations
we have defined. For instance, we have
$$
\app^*(\lambda^*(b,b'),(a,a')) = (\app(\lambda b,a),\app(\lambda b,a)') = (b(1,a),\app(\lambda b,a)')
$$
and
$$
\app(\lambda b,a)'\rho\rho' = (\lambda b)'\rho\rho'(a\rho)(a'\rho\rho') = b'(\rho,a\rho)(\rho',a'\rho\rho')
$$
and
$$
(b(1,a))'\rho\rho' = b'(\rho,a\rho)(1'\rho\rho',a'\rho\rho') = b'(\rho,a\rho)(\rho',a'\rho\rho')
$$
When checking the equalities, we {\em only use $\beta,\eta$-conversions at the metalevel}.

\medskip

 There are of course strong similarities with the parametricity model presented in \cite{BJ}.
This model can also be seen as a constructive version of the {\em glueing} technique \cite{LS,Shulman}.
Indeed, to give a family of sets over $|\Gamma|$ is essentially the same as to give a set $X$ and a map
$X\rightarrow |\Gamma|$, which is what happens in the glueing technique \cite{LS,Shulman}.
%, which can be traced back to Skolem \cite{Skolem}.

\section{The term model}

 There is a canonical notion of morphism between two models.
For instance,  the first projection $\MM^*\rightarrow \MM$ defines a map of models of type theory. 
As for models of generalized algebraic theories 
\cite{Dybjer}, there is an {\em initial} model unique up to isomorphism.
We define the {\em term} model $\MM_0$ of type theory to be this initial model.
As for equational theories, this model can be presented by first-order terms
(corresponding to each operations) modulo the equations/conversions that have to
hold in any model.
 
\begin{theorem}
In the initial model given $u$ in $|N_2|$ we have $u = 0$ or $u = 1$. Furthermore
we don't have $0 = 1$ in the initial model.
\end{theorem}

\begin{proof}
 We have a unique
map of models $\MM_0\rightarrow \MM_0^*$. The composition of the first projection 
with this map has to be the identity function on $\MM_0$.
If $u$ is in $|N_2|$ the image of $u$ by the initial map has hence to be a pair of the
form $u,u'$ with $u'$ in $N_2'(u)$. It follows that we have $u = 0$ if $u' = 0$
and $u = 1$ if $u' = 1$. Since $0' = 0$ and $1' = 1$ we cannot have $0 = 1$ in 
the initial model $\MM_0$.
\end{proof}

\section{Presheaf model}

 We suppose given an arbitrary model $\MM$. We define from this the following
category ${\cal C}$ of ``telescopes''. An object of ${\cal C}$ is a list
$A_1,\dots,A_n$ with $A_1$ in $\Type()$, $A_2$ in $\Type(A_1)$,
$A_3$ in $\Type(A_1.A_2)$ $\dots$ To any such object $X$ we can associate a context
$i(X) = A_1.\dots.A_n$ of the model $\MM$. If $A$ is in $\Type(i(X))$, we define
the set $\Var(X,A)$ of numbers $v_k$ such that $\qq\pp^{n-k}$ is in $\Elem(i(X),A)$.
We may write simply $\Elem(X,A)$ instead of $\Elem(i(X),A)$. Similarly
we may write $\Type_n(X) = \Elem(X,U_n)$ for $\Type_n(i(X))$.
 If $v_k$ is in
$\Var(X,A)$ we write $[v_k] = \qq\pp^{n-k}$. If $Y = B_1,\dots,B_m$ is an
object of ${\cal C}$, a map $\sigma:Y\rightarrow X$ is given by a list
$u_1,\dots,u_n$ such that $u_p$ is in $\Var(Y,A_p([u_1],\dots,[u_{p-1}]))$.
We then define $[\sigma] = ([u_1],\dots,[u_p]):i(Y)\rightarrow i(X)$.
It is direct to define a composition operation such that $[\sigma\delta] = [\sigma][\delta]$
which gives a category structure on these objects.

 We use freely that we can interpret the language of dependent types (with universes)
in any presheaf category \cite{Hofmann1}.
A presheaf $F$ is given by a family of sets $F(X)$
indexed by contexts with restriction maps $F(X)\rightarrow F(Y),~u\mapsto u\sigma$
if $\sigma:Y\rightarrow X$, satisfying the equations
$u1 = u$ and $(u\sigma)\delta = u(\sigma\delta)$ if $\delta:Z\rightarrow Y$.
A dependent presheaf $G$ over $F$ is a presheaf over
the category of elements of $F$, so it is given by a family of sets $G(X,\rho)$
for $\rho$ in $F(X)$ with restriction maps. 

  We write $\GV_0,\GV_1,\dots$ the cumulative sequence
of presheaf universes, so that $\GV_n(X)$ is the set of $\GU_n$-valued
dependent presheaves on the presheaf represented by $X$.

 $\Type_n$ defines a presheaf over this category, with $\Type_n$ subpresheaf
of $\Type_{n+1}$. We can see $\Elem$ as a dependent presheaf over $\Type_n$
since it determines a collection of sets $\Elem(X,A)$ for $A$ in $\Type_n(X)$
with restriction maps.
%We let $\TNorm$ (resp. $\TNeut$) be the set of expressions that are types in normal
%form (resp. neutral types). 
 
%% We let $\Term(A)$ be the set of expressions of type $A$, which is well-defined since
%% we can apply a renaming to such a term. This also define a dependent
%% presheaf over $\Type_n$. We have a ``value'' function $\Term(A)\rightarrow\Elem(A),~e\mapsto \Val{e}$
%% which corresponds to the map $\Term(\Gamma,A)\rightarrow\Elem(\Gamma,A)$
%% which associates to a term its equivalence class modulo conversion.

 If $A$ is in $\Type_n(X)$ we let $\Norm(X,A)$ (resp. $\Neut(X,A)$)
be the set of all expressions of type $A$ that are in normal
form (resp. neutral). As for $\Elem$, we can see $\Neut$ and $\Norm$
as dependent types over $\Type_n$, and we have 
$$\Var(A)\subseteq \Neut(A)\subseteq\Norm(A)$$
%All these can be seen as functions in $\Type_n\rightarrow\GV_0$.
 We have an evaluation function $[e]:\Elem(A)$ if $e:\Norm(A)$.
If  $a$ is in $\Elem(A)$ then we let $\Norm(A)|a$ (resp. $\Neut(A)|a$) be
the subtypes of $\Norm(A)$ (resp. $\Neut(A)$) of elements $e$ such that $[e] = a$.

\medskip

 Each context $\Gamma$ defines a presheaf $|\Gamma|$ by letting $|\Gamma|(X)$ be
the set of all substitutions $i(X)\rightarrow\Gamma$. 

%Note that this is {\em not} the presheaf
%represented by $\Gamma$, since we take arbitrary substitutions and  not 
%term renaming substitutions.

%\subsection*{Higher-order abstract syntax}

% This presheaf model can be used for higher-order syntax, similarly to the models
%used in the reference \cite{AHS}.

 Any element $A$ of $\Type_n(\Gamma)$ defines internally a function
$|\Gamma|\rightarrow\Type_n,~\rho\mapsto A\rho$.

We have a canonical isomorphism between $\Var(A)\rightarrow\Type_n$ and $\Elem(A\rightarrow U_n)$.
We can then use this isomorphism to build an operation 
$$
\pi : \Pi (A:\Type_n)(\Var(A)\rightarrow\Type_n)\rightarrow \Type_n
$$
such that $(\Pi~A~B)\rho = \pi (A\rho) ((\lambda x:\Var(A\rho))B(\rho,[x]))$.

 We can also define, given $A:\Type_n$ and $F:\Var(A)\rightarrow\Type_n$
an operation $\Lambda A f:\Elem(\pi A F)$, for $f:\Pi (x:\Var(A))\Elem(F~x)$.

 Similarly, we can define an operation
$$
\pi : \Pi (A:\Norm(\UU_n))(\Var([A])\rightarrow\Norm(\UU_n))\rightarrow \Norm(\UU_n)
$$
such that $[\pi A F] = \pi [A] (\lambda (x:\Var([A]))[F~x])$ and 
given $A:\Norm(\UU_n)$ and $F:\Var([A])\rightarrow\Type_n$
and $f:\Pi (x:\Var([A]))\Elem(F~x)$
an operation $\Lambda A f:\Norm(\pi [A] F)$ such that
$[\Lambda A f] = \Lambda [A] (\lambda (x:\Var([A])[f~x]))$.

 While equality might not be decidable in $\Var(A)$ (because we use arbitrary renaming
as maps in the base category), the product operation is injective: if
$\pi A F = \pi B G$ in $\Norm(\UU_n)$ then $A = B$ in $\Norm(\UU_n)$ and
$F = G$ in $\Var([A])\rightarrow\Type_n$.

\section{Normalization model}

 The model is similar to the reducibility model and we only explain the main operations.

 As before, a context is a pair $\Gamma,\Gamma'$ where $\Gamma$ is a context of $\MM$
and $\Gamma'$ is a dependent family over $|\Gamma|$.

\medskip

 A type at level $n$ over this context 
consists now of a pair $A,\lift{A}$ where 
$A$ is in $\Type_n(\Gamma)$ and 
$\lift{A}\rho\rho'$ in $U_n'(A\rho)$
for $\rho$ in $|\Gamma|$ and $\rho'$ in $\Gamma'(\rho)$.
An element of $U_n'(T)$ for $T$ in $\Type_n$ consists in 
a 4-uple $T',T_0,\alpha,\beta$
where the element $T_0$ is in $\Norm(U_n)|T$,
the element $T'$ is in $\Elem(T)\rightarrow \GV_n$,
the element 
$\beta$ is in $\Pi (k :  \Neut(T))T'(\Val{k})$
and $\alpha$ is in $\Pi (u : \Elem(T))~T'(u)\rightarrow \Norm(T)|u$.

\medskip

 An element of this type is a pair $a,\lift{a}$ where $a$ is in $\Elem(\Gamma,A)$
and $\lift{a}\rho\rho'$ is an element of $T'(a\rho)$ where
$(T',T_0,\alpha,\beta) = \lift{A}\rho\rho'$.

\medskip

 The intuition behind this definition is that it is a ``proof-relevant'' way
to express the method of reducibility used for proving normalization \cite{FLD}: a 
reducibility predicate has to contain all neutral terms and only normalizable terms.
The function $\alpha$ (resp. $\beta$) is closely connected to the ``reify'' (resp. ``reflect'') 
function used in normalization by evaluation \cite{BS}, but for a ``glued'' model.
%like it was done for combinators in \cite{CD}.

\medskip

 We redefine ${N_2}'(t)$ to be the set of elements
$u$ in $\Norm(N_2)|t$
such that $u$ is $0$ or $1$ or is neutral.
We define $\alpha_{N_2} t \nu = \nu$ and $\beta_{N_2}(k) = k$.
%and $\lift{N_2} = (N_2',N_2,\alpha_{N_2},\beta_{N_2}).$

\medskip

 We define $\alpha_{\UU_n}~T~(T',T_0,\alpha_T,\beta_T) = T_0$ 
and for $K$ neutral $\beta_{U_n}(K) = (K',K,\alpha,\beta)$ where $K'(t)$ is
$\Neut(\Val{K})|t$ and $\alpha t k = k$ and $\beta(k) = k$.

\medskip

 The set $\Type^*_n(\Gamma,\Gamma')$ is defined to be the set of pairs
$A,\lift{A}$ where $A$ is in $\Type_n(\Gamma)$ and
$\lift{A}\rho\rho'$ is in $U'_n(A\rho)$.

\medskip

 The extension operation is defined by $(\Gamma,\Gamma').(A,\lift{A}) = \Gamma.A,(\Gamma.A)'$
where $(\Gamma.A)'(\rho,u)$ is the set of pairs $\rho',\nu$
with $\rho'\in \Gamma'(\rho)$ and $\nu$ in $\lift{A}\rho \rho'.1(u)$.

\medskip

 We define a new operation $\Pi^*~(A,\lift{A})~(B,\lift{B}) = C,\lift{C}$ where $C = \Pi~A~B$
and $\lift{C}\rho\rho'$ is the tuple
%$\Pi~E_A~E_B = (C',C_0,\alpha,\beta)$ where
\begin{itemize}
\item $C'(w) = \Pi (a : \Elem(A\rho)) \Pi (\nu :  T'(u))F' u \nu(\APP{w}{u})$
\item $\beta(k) u \nu = \beta_F u \nu (\APP{k}{\alpha_T u \nu})$
\item $\alpha~ w~\xi =
         \Lambda T_0 g$ with $g(x) = \alpha_F \Val{x} \beta_T(x) (\APP{w}{\Val{x}}) (\xi \Val{x} \beta_T(x))$
\item $C_0 = \pi T_0 G$ with $G(x)  = F_0 \Val{x} \beta_T(x)$
\end{itemize}
where we write $(T',T_0,\alpha_T,\beta_T) = \lift{A}\rho\rho'$ in $U_n'(A\rho)$
and for each $u$ in $\Elem(A\rho)$ and $\nu$ in $T'(u)$ we write
$(F'u \nu,F_0 u \nu,\alpha_F u \nu,\beta_F u \nu) = \lift{B}(\rho,u)(\rho',\nu)$
in  $U_n'(B(\rho,u))$. We can check $[C_0] = (\Pi~A~B)\rho$
and we have $C',C_0,\alpha,\beta$ is an element in $U_n'((\Pi~A~B)\rho).$
%$, which defines a
%dependent product operation in this new model, using Lemma \ref{corr} to check
%the correctness of $\alpha$ and $C_0$.

\medskip

 We define $\lift{U_n} = U_n,{U_n}',\alpha_{U_n},\beta_{U_n}$ and 
$\lift{N_2} = N_2,{N_2}',\alpha_{N_2},\beta_{N_2}$.

\medskip

 If we have $T$ in $\Type_n(\Gamma.N_2)$ and $a_0$ in $\Elem(T\subst{0})$ and $a_1$ in $\Elem(T\subst{1})$
and for each $\rho:|\Gamma|$ and $\rho':\Gamma'(\rho)$ and
$u$ in $\Elem(N_2)$ and $\nu$ in $N_2'(u)$ an element
$(T'u \nu,T_0 u \nu,\alpha_T u \nu,\beta_T u \nu)$
in  $U_n'(T(\rho,u))$ and $\lift{a_0}$ in $T'00({a_0})$ and
$\lift{a_1}$ in $T'11{a_1}$ we define $f = \lift{\brec(T,a_0,a_1)}\rho\rho'$ as follows.
We take $f~u~\nu = \lift{a_0}$ if $\nu = 0$ and
$f~u~\nu = \lift{a_1}$ if $\nu = 1$ and finally
$f~u~\nu = 
 \beta_T u \nu (\brec(\Lambda(N_2,g),\alpha_T 0 0 {a_0} \lift{a_0},
                \alpha_T 1 1 {a_1} \lift{a_1}))(\nu))$
where $g(x) = T_0\Val{x}\beta_{N_2}(x)$ if $\nu$ is neutral.

\medskip

 We thus get, starting from an arbitrary model $\MM$, a new model $\MM^*$ with a projection
map $\MM^*\rightarrow \MM$. As for the canonicity model, if we start from the initial model $\MM_0$
we have an initial map $\MM_0\rightarrow\MM_0^*$ which is a section of the projection
map. Hence for any $a$ in $\Elem(A)$ we can compute $\lift{a}$ in $A'(a)$
where $(A',A_0,\alpha_A,\beta_A) = \lift{A}$ and we have $\alpha_A~a~\lift{a}$
in $\Norm(A)|a$.

\begin{theorem}
Equality in $\MM_0$ is decidable.
\end{theorem}

\begin{proof}
If $a$ and $b$ are of type $A$
we can compute $\lift{A} = (A',A_0,\alpha,\beta)$. We then have $a = b$ in $\Elem(A)$
if, and only if, $\alpha a \lift{a} = \alpha b \lift{b}$ in $\Norm(A)$
since $u = [\alpha u \lift{u}]$ for any $u$ in $\Elem(A)$. The result then
follows from the fact that the equality in $\Norm((),A)$ is decidable.
\end{proof}

 We also can prove that $\Pi$ is one-to-one for conversions, following P. Hancock's argument
presented in \cite{ML73}.

%% \begin{corollary}
%% If $ \Pi (x:A_0) B_0 ~\conv~ \Pi (x:A_1) B_1$ then $A_0 ~\conv~ A_1$ and
%% $B_0 ~\conv~ B_1~(x:A)$.
%% \end{corollary}

\section{Conclusion}

 Our argument extends directly to the addition of dependent sum types with surjective
pairing, or inductive types such as the type $\WW~A~B$ \cite{ML79}.

 The proof is very similar to the argument presented in \cite{ML73}, but it covers
conversion under abstraction and $\eta$-conversion. Instead of set theory, one could formalize
the argument in extensional type theory; presheaf models have been already represented elegantly
in NuPrl \cite{Bickford}.
As we  noticed however, the meta theory only uses the form of extensionality ($\eta$-conversion)
also used in the object theory, and we should be able to express the normalization proof as
a program transformation from one type theory to another. The formulation of the presheaf
model as a(n extension of) type theory will be similar to the way cubical type theory \cite{CCHM} expresses syntactically
a presheaf model over a base category which is a Lawvere theory. This should amount essentially
to work in a type theory with a double context, where substitutions for the first context are
restricted to be renamings. We leave this as future work,
which, if successful, would refute some arguments in \cite{ML74} for not accepting $\eta$-conversion 
as definitional equality.

\section*{Acknowledgement}

 This work started as a reading group of the paper \cite{Shulman} together with Simon Huber and
Christian Sattler. The discussions we had were essential for this work; in particular Christian
Sattler pointed out to me the reference \cite{AHS}

\end{document}